\documentclass[10 pt, conference]{ieeeconf}

\IEEEoverridecommandlockouts 
\makeatletter

\let\proof\@undefined
\let\endproof\@undefined
\makeatother
\let\labelindent\relax

\usepackage{mathtools,amssymb,amsthm,bm}
\usepackage[psamsfonts]{eucal}
\usepackage[ruled,vlined,linesnumbered]{algorithm2e}
\usepackage{enumitem,setspace,microtype,threeparttable,float,afterpage,placeins,cancel,leftidx,breqn}
\usepackage[export]{adjustbox}
\usepackage[dvipsnames]{xcolor}
\usepackage[most]{tcolorbox}
\usepackage{soul}
\usepackage{cite}
\usepackage{standalone}
\usepackage{fancyhdr}

\usepackage{jhoagg}

\definecolor{algorithmcomment}{RGB}{0,0,255}

\SetKwInput{KwInput}{Input}
\SetKwInput{KwOutput}{Output}
\SetKwComment{Comment}{// }{}
\SetCommentSty{AlgorithmCommentStyle}

\usepackage[hidelinks]{hyperref}
\usepackage[capitalize,noabbrev]{cleveref}

\newcounter{appsection}
\renewcommand{\theappsection}{\Alph{appsection}}

\crefname{appsection}{Appendix}{Appendices}
\Crefname{appsection}{Appendix}{Appendices}

\newcommand{\appsection}[2][]{%
    \refstepcounter{appsection}%
    \section*{\texorpdfstring{APPENDIX \theappsection \\ #2}{Appendix #2}}%
    \addcontentsline{toc}{section}{Appendix \theappsection: #2}%
    \ifx&#1&\else\label{#1}\fi
}

\let\originalleft\left
\let\originalright\right
\renewcommand{\left}{\mathopen{}\mathclose\bgroup\originalleft}
\renewcommand{\right}{\aftergroup\egroup\originalright}

\fancypagestyle{firststyle}
{\fancyfoot{} 
\fancyfoot[L]{\textit{Preprint submitted to 2024 American Control Conference (ACC)}}
 
}

\newtheoremstyle{indentedplain}
  {3pt}      
  {3pt}      
  {\rmfamily}
  {1em}      
  {\bfseries}
  {.}        
  {0.5em}    
  {}         

\newtheoremstyle{indentedremark}%
  {3pt}{3pt}{}
  {1em}{\bfseries}{.}{0.5em}{}

\theoremstyle{indentedplain}
\newtheorem{theorem}{Theorem}

\newtheorem{proposition}{Proposition}

\theoremstyle{indentedremark}

\allowdisplaybreaks

\newcounter{enumA}
\newcounter{enumC}
\newcounter{enumO}

\newlist{enumA}{enumerate}{1}
\setlist[enumA,1]{%
  label=(A\arabic*),
  leftmargin=1cm,
  before=\setcounter{enumAi}{\value{enumA}},
  after=\setcounter{enumA}{\value{enumAi}}
}

\newlist{enumC}{enumerate}{1}
\setlist[enumC,1]{%
  label=(C\arabic*),
  leftmargin=1cm,
  before=\setcounter{enumCi}{\value{enumC}},
  after=\setcounter{enumC}{\value{enumCi}}
}

\newlist{enumO}{enumerate}{1}
\setlist[enumO,1]{%
  label=(O\arabic*),
  leftmargin=1cm,
  before=\setcounter{enumOi}{\value{enumO}},
  after=\setcounter{enumO}{\value{enumOi}}
}

\crefname{equation}{}{}

\newcommand{\ubar}[1]{{\underline{#1}}}
\DeclareMathOperator*{\argmin}{argmin}

\DeclareMathAlphabet{\mathcal}{OMS}{cmsy}{m}{n} 

\begin{document}
\title{Safe Receding-Horizon Control of Fixed-Wing Aircraft Using Constrained Approximate Dynamic Programming}
\author{Felipe Arenas-Uribe, Ricardo Gutierrez, and Jesse B. Hoagg
\thanks{F. Arenas-Uribe, R. Gutierrez, and Jesse B. Hoagg are with the Department of Mechanical and Aerospace Engineering, University of Kentucky, Lexington, KY, USA. (e-mail: felipearur@uky.edu, ricardo.gutierrez@uky.edu, jesse.hoagg@uky.edu).}
\thanks{This work is supported in part by the National Science Foundation (2450718, 2531196) and USDA NIFA (2024-69014-42393).}       }                                                                                                                                                                                                                                                                    
\maketitle


\begin{abstract}
We present a receding-horizon optimal control for fixed-wing aircraft subject to state constraints. 
These constraints include operational constraints such as altitude, geofencing, and obstacle avoidance, as well as bounds on flight-path angle, roll angle, and airspeed. 
The state constraints are composed into a single soft-minimum control barrier function (CBF). 
We use this composite CBF in a constrained-approximate dynamic program to obtain a sequence of analytic closed-form control functions that approximately minimize a quadratic finite-horizon integral cost subject to the CBF constraint. 
The resulting receding-horizon control is non-myopic in the sense that it approximately optimizes the integral cost while satisfying the state constraint at all times along the entire prediction horizon.
We demonstrate constraint satisfaction and performance in simulations of a fixed-wing aircraft navigating an obstacle-filled airspace under wind uncertainty. 
We also compare this receding-horizon control with 2 other methods.
\end{abstract}

\section{Introduction}

Over the past decade, uncrewed air vehicles (UAVs) have grown rapidly in deployment and applications. 
Fixed-wing UAVs are of interest, in large part, because of their energy efficiency,  flight range, and cargo capabilities \cite{beard_small_2012}.
Applications of fixed-wing UAVs include wild-fire sensing, distributed atmospheric sensing, surveillance and precision agriculture \cite{kim_wildland_2026,sinha_utilizing_2024,zheng_constrained_2023}.
As UAVs continue to become more common, there is a need for autonomous and safe operation within the airspace.

Standard approaches for fixed-wing aircraft control are based on control loops which regulate the deflection of control surfaces using PID or LQR controllers \cite{nelson_flight_1998}. 
As mission requirements become more complex, more state-of-the-art control methods can provide better performance \cite{Dydek_NASA_2010,stastny_nonlinear_2017,heintz_formation_2023}. 
However, these control approaches do not directly handle safety and actuator constraints. 
Thus, constrained control for fixed-wing UAVs remains an important area of research \cite{Cohen_Constrained_2026}.

Constrained control can be formalized as forward invariance of a set of states, referred to as a safe set. 
Methods that address this problem include MPC, reference governors, reachability analysis, and control barrier functions (CBFs). 
In particular, CBFs have been of interest for aerospace applications due to their rigorous guarantees and computational efficiency \cite{Cohen_Constrained_2026}.
Application of CBF methods as safety-filters for fixed-wing aircraft have been addressed in \cite{squires_composition_2022,zheng_constrained_2023,c_altunkaya_nuisance-free_2025,molnar_collision_2025}. 
However, this CBF-based optimization is myopic in the sense that the constraint is enforced pointwise in time based only on the current state without considering how the system trajectory will evolve over a future horizon.

To address this shortcoming, non-myopic CBF methods rely on prediction over a horizon. 
In \cite{breeden_predictive_2022}, the trajectory predicted under a nominal control is used to enforce safety over the horizon; however, the control does not optimize performance over that horizon.
Layered CBF architectures \cite{Sforni_Receding-Horizon_CBF} and nonlinear MPC with CBF constraints \cite{Huang_MPCBF_2025,Zeng_NMPC_2021} balance safety and performance but can be computationally expensive.
Alternatively, predicted-flow CBFs can be used to certify safety over the prediction horizon while simultaneously allowing for performance optimization over the horizon \cite{Safari_PCBFS_OJC_2026,safari_PCBFS_2026}. 
Another approach is presented in \cite{gutierrez_receding-horizon_2026} using constrained approximate dynamic programming (C-ADP) and a quadratic approximation of the cost-to-go. 
This method uses C-ADP combined with a CBF-like constraint to obtain a receding-horizon optimal control that approximately optimizes an integral cost, while guaranteeing that state constraints are satisfied at all times along the entire prediction horizon.

This article presents a receding-horizon optimal control for fixed-wing aircraft subject to state constraints. 
The controller uses a Dubins aircraft model \cite{valavanis_implementing_2015}, where the state is extended to include aircraft speed, roll, and flight path. 
In other words, the control inputs are speed rate, roll rate, and flight-path rate. 
We use barrier functions to encode constraints on flight-path angle, roll angle, airspeed, altitude, geofencing, and obstacle avoidance. 
A wind estimate and an upper bound on the estimation error are used to construct robust airspeed barrier functions. 
We compose the barrier functions into a single soft-minimum CBF, which produces a constraint that can guarantee forward invariance of a subset of the safe set \cite{Rabiee_2024_CBFs,Rabiee_2023_Composition}. 
Next, we use the CBF with C-ADP to construct a sequence of analytic closed-form control functions that approximately minimize a finite-horizon integral cost while satisfying the state constraint along the entire prediction horizon. 
We demonstrate constraint satisfaction and performance in simulations of a fixed-wing aircraft navigating an obstacle-filled airspace under uncertain wind, and we compare the proposed control with two alternative methods.

\section{Notation}

Physical vectors are denoted with arrows (e.g., $\vec{r}$). 
A frame is a collection of mutually orthogonal physical unit vectors. 
If $\SF$ is a frame, then $[\vec{r}]_{\SF}$ is $\vec{r}$ resolved in $\SF$. 
The time derivative of $\vec{r}$ with respect to the frame $\SF$ is denoted by $^{\SF}\dot{\vec{r}}$.


Let $I_3$ be the $3 \times 3$ identity matrix, for $j \in \{1,2,3\}$, let $e_j \in \BBR^3$ be the $j$-th column of $I_3$.
Let $\| \cdot \|$ denote the 2-norm. 
The boundary of $\mathcal{A} \subseteq \BBR^n$ is denoted by $\text{bd }\mathcal{A}$. 


\section{Problem Formulation}

Let $\SF_{\rm E}$ be an inertial frame. The origin of $\SF_{\rm E}$ is $\rm{o}_{\rm E}$ which is any convenient point on the surface of the Earth.
Let $\rm{o}_{\rm A}$ be the location of the aircraft.
The position of $\rm{o}_{\rm A}$ relative to $\rm{o}_{\rm E}$ is $\vec{q}$, which is resolved in $\SF_{\rm E}$ as $q \triangleq \vec{q}|_{\SF_{\rm E}}$.
The velocity of $\rm{o}_{\rm A}$ relative to $\rm{o}_{\rm E}$ with respect to $\SF_{\rm E}$ is $\vec{p} \triangleq \, ^{\SF_{\rm E}}\dot{\vec{q}}$.

Let $\SF_{\rm B}$ be the body frame such that $\vec{p}$ resolved in $\SF_{\rm B}$ is $\vec{p}|_{\SF_{\rm B}} = s e_1$, where $s:[0,\infty) \to (0,\infty)$ is the speed of the aircraft.
Let $\psi:[0,\infty) \to \BBR$ be the aircraft heading angle, and let $\gamma: [0,\infty) \to (-\pi/2, \pi/2)$ be the aircraft flight-path angle. 
The pointing direction is given by
\begin{equation}\label{eq:pointing_direction}
    y\big(\psi,\gamma\big) \triangleq \begin{bmatrix}
        \cos \gamma  \cos \psi \\
        \cos \gamma  \sin \psi \\
        -\sin \gamma
    \end{bmatrix}.
\end{equation}
Note that $y(\psi,\gamma)$ is the unit vector in the direction of $\dot{q}$. 
Thus, the velocity $\vec{p}$ resolved in $\SF_{\rm E}$ is $\vec{p}|_{\SF_{\rm E}} = s y(\psi, \gamma)$, which implies that
\begin{equation}\label{eq:position_dynamics}
    \dot{q}(t) = s(t) y\big(\psi(t),\gamma(t)\big),
\end{equation}
where $q(0) = q_0 \in \BBR^3$. 
The heading angle of the aircraft satisfies
\begin{equation}\label{eq:heading_angle_dynamics}
    \dot{\psi}(t) = \frac{g_{\rm E}}{s(t) \cos \gamma(t)}\tan \phi(t),
\end{equation}
where $\psi(0) = \psi_0 \in \BBR$, $g_{\rm E} = 9.81$~m/s$^2$ and $\phi: [0,\infty) \to (-\pi/2,\pi/2)$ is the roll angle.

Equations~\cref{eq:pointing_direction,eq:position_dynamics,eq:heading_angle_dynamics} are the {Dubins aircraft model} \cite{valavanis_implementing_2015}, which is a kinematic model of a fixed-wing aircraft. 
However, inner-loop aircraft control often specifies the time derivatives of $s$, $\phi$ and $\gamma$, that is, speed rate, roll rate, and flight-path rate \cite{nelson_flight_1998}. 
Thus, we let $s$, $\phi$, and $\gamma$ satisfy 
\begin{equation}\label{eq:control_dynamics}
    \begin{bmatrix} \dot{s}(t) \\ \dot{\gamma}(t) \\ \dot{\phi}(t) \end{bmatrix} = u(t), 
\end{equation}
where $s(0) = s_0 > 0$, $\phi(0) = \phi_0 \in (-\pi/2,\pi/2)$, $\gamma(0) = \gamma_0 \in (-\pi/2,\pi/2)$, and $u \colon [0,\infty) \to \BBR^3$ is the control.


The aircraft model \cref{eq:pointing_direction,eq:position_dynamics,eq:heading_angle_dynamics,eq:control_dynamics} can be written as
\begin{equation} \label{eq:system_dynamics}
    \dot{x}(t) = f(x(t)) + Bu(t),
\end{equation}
where $n=7$ and 
\begin{gather}
x \triangleq \left [\begin{smallmatrix} q \\ 
        \psi \\ 
        s \\ 
        \gamma \\ 
        \phi
    \end{smallmatrix} \right ], \,\,
x_0 \triangleq \left [\begin{smallmatrix}
        q_0 \\ 
        \psi_0 \\
        s_0 \\ 
        \gamma_0 \\ 
        \phi_0
    \end{smallmatrix} \right ],\,\,
f(x) \triangleq 
    \begin{bmatrix}
        s y(\psi, \gamma) \\
        \frac{g_{\rm E} \tan \phi}{s \cos \gamma} \\
        0_{3 \times 1}
    \end{bmatrix}, \,\,
B \triangleq 
    \begin{bmatrix}
        0_{4 \times 3} \\ I_{3 \times 3}
    \end{bmatrix}.\label{eq:system_dynamics.2}
\end{gather}
Note that \cref{eq:system_dynamics,eq:system_dynamics.2} is well defined on
\begin{equation*}
    \SD \triangleq \{ x \in \BBR^n \colon s > 0, -{\pi}/{2} <\gamma<{\pi}/{2}, -{\pi}/{2} <\phi<{\pi}/{2}\}.
\end{equation*}

The flight-path angle $\gamma$ is required to be in the interval $[ \ubar{\gamma}, \bar{\gamma} ]$, where $-\frac{\pi}{2} < \ubar{\gamma} < \bar{\gamma} < \frac{\pi}{2}$, and the roll angle $\phi$ is required to be in the interval $[ \ubar{\phi}, \bar{\phi} ]$, where $-\frac{\pi}{2} < \ubar{\phi} < \bar{\phi} < \frac{\pi}{2}$.
These constraints are modeled with the barrier functions $h_1,\ldots,h_4 \colon \BBR^n \to \BBR$ given by 
\begin{align}
    h_1(x) &\triangleq \gamma - \ubar{\gamma}, \quad h_2(x) \triangleq \bar{\gamma} - \gamma, \label{eq:flight_path_constraint} \\
    h_3(x) &\triangleq \phi - \ubar{\phi}, \quad h_4(x) \triangleq \bar{\phi} - \phi. \label{eq:roll_constraint}
\end{align}

Let $\vec{w}$ be the velocity of the air mass around the aircraft relative to $\rm{o}_{\rm E}$ with respect to $\SF_{\rm E}$, and define $w \triangleq \vec{w}|_{\SF_{\rm E}}$. 
Thus, the airspeed of the aircraft is $v(t,x) \triangleq \| s y(\psi,\gamma) - w(t)\|$, which is required to be in the interval $[ \ubar{v}, \bar{v} ]$, where $\bar{v} > \ubar{v} > 0$.
These constraints are modeled with the barrier functions $h_5,h_6 \colon [0,\infty) \times \BBR^n \to \BBR$ given by 
\begin{align}
    h_5(t,x) &\triangleq v(t,x)^2 - \ubar{v}^2 = \| sy(\psi,\gamma) - w(t) \|^2 - \ubar{v}^2, \label{eq:airspeed_cons_1}\\
    h_6(t,x) &\triangleq \bar{v}^2 - v(t,x)^2 = \bar{v}^2 - \| sy(\psi,\gamma) - w(t) \|^2. \label{eq:airspeed_cons_2}
\end{align}
Note that $v(t,x) \in [ \ubar{v}, \bar{v} ]$ if and only if $h_5(t,x) \ge 0$ and $h_6(t,x) \ge 0$.

The aircraft is also subject to constraints on position $q$, and these constraints can vary in time.
For example, there can be geofence constraints, static and dynamic obstacles, and altitude constraints. 
These constraints are modeled by the intersection of the zero-superlevel sets of barrier functions $b_1,\ldots,b_r: [0,\infty) \times \BBR^3 \to \BBR$.
Thus, for each $t \ge 0$, the set of states that satisfy all constraints is
\begin{align*}
    \mathcal{S}(t) &\triangleq \{ x \in \BBR^n : h_1(x) \geq 0, \ldots ,h_4(x) \geq 0,  h_5(t,x) \geq 0, \\
    &\qquad  h_6(t,x) \geq 0, b_1(t,q) \geq 0, \ldots, b_r(t,q) \geq 0\}.
\end{align*}

Next, let $T>0$ be a prediction horizon, and for each $t\geq 0$, consider the receding-horizon cost
\begin{align}
J \left ( t,x_0,\bar{u}(\cdot) \right )
    &\triangleq \int_t^{t+T} 
    \frac{1}{2}x(\tau)^\top Q(\tau)x(\tau) +\Gamma(\tau)^\top x(\tau) \nn\\
    &\quad +\frac{1}{2}\bar{u}(\tau)^\top R(\tau)\bar{u}(\tau) +\Omega(\tau)^\top\bar{u}(\tau) \, \mathrm{d}\tau,\label{eq:cost_function}
\end{align}
where $x$ satisfies \cref{eq:system_dynamics,eq:system_dynamics.2} with $u = \bar{u}$, $Q(t) \in \BBR^{n  \times n}$ is positive semidefinite, $R(t) \in \BBR^ {3 \times 3}$ is positive definite, $\Gamma(t) \in \BBR^{n}$, and $\Omega(t) \in \BBR^3$.

The wind $w$ is not assumed to be known. 
Instead, let $\hat{w} \colon [0, \infty) \to \BBR^3$ be a continuously differentiable estimate of wind velocity, where at time $t\ge0$, $\hat{w}$ and $\dot{\hat{w}}$ are assumed to be known on the interval $[t, t + T]$.
We assume that there is a known upper bound on the magnitude of the estimation error.
Specifically, there exists a known $\bar w >0$ such that for all $t \ge 0$, $\|w(t) - \hat{w}(t)\| \le \bar{w}$.
For example, $\hat w$ could be an estimate of the mean wind, and $\bar w$ an upper bound on the magnitude of the wind perturbation from its mean. 

The objective of this article is to construct a control $u$ that minimizes the cost \cref{eq:cost_function} subject to the constraint that $x(t) \in \mathcal{S}(t)$ for all $t \in [0,\infty)$.

\section{Soft-Minimum Relaxed CBF}

This section presents a constraint that can make a subset of $\mathcal{S}(t)$ forward invariant.
This constraint is based on a relaxed CBF, that is, a function that satisfies CBF criteria on the boundary of its zero-superlevel set but not necessarily on the interior. 

We note that $h_1, \dots h_6$ have relative degree~1 with respect to \cref{eq:system_dynamics,eq:system_dynamics.2}, whereas $b_1, \ldots, b_r$ have relative degree~2 with respect to \cref{eq:system_dynamics,eq:system_dynamics.2}. 
Thus, we adopt a higher-order approach to construct candidate CBFs from $b_1, \ldots, b_r$.
For $i \in \{1, \ldots, r\}$, let $\beta_{i} \colon \BBR \to \BBR$ be a 2-times continuously differentiable extended class-$\mathcal{K}$ function. 
Then, consider $h_{i+6} \colon [0,\infty) \times \BBR^{n} \to \BBR$ defined by
\begin{equation}
h_{i+6}(t, x) \triangleq \frac{\partial b_i(t,q)}{\partial t} + \frac{\partial b_i(t,q)}{\partial q} s y(\psi,\gamma) + \beta_{i}(b_{i}(t,q)).\label{eq:h_11}
\end{equation}
For all $t \ge 0$, the intersection of zero-superlevel sets of $h_1,\ldots,h_{r+6}$ is 
\begin{align*}
    \mathcal{P}(t) &\triangleq \{ x \in \BBR^n : h_1(x) \geq 0, \ldots ,h_4(x) \geq 0, \\
    &\qquad  h_5(t,x) \geq 0, \ldots, h_{r+6}(t,x) \geq 0\}.
\end{align*}
and define
\begin{equation}
\SC(t) \triangleq \SSS(t) \cap \SP(t). \label{eq:C_set}
\end{equation}
Note that \cite[Theorem 3]{xiao_high-order_2022} implies that if $x_0 \in \SC(0)$ and for all $t \in [0, \infty)$, $x(t) \in \SP(t)$, then for all $t \in [0, \infty)$, $x(t) \in \SC(t) \subseteq \mathcal{S}(t)$.
However, $\SP$ is constructed from the zero-superlevel set of \cref{eq:airspeed_cons_1,eq:airspeed_cons_2}, which are not implementable because they depend on the unknown wind $w$. 
Thus, we use the estimate $\hat w$ and upper bound $\bar w$ to construct implementable lower bounds. 
In particular, let $\varepsilon> 0$, and consider $\ubar{h}_{5},\ubar{h}_{6} \colon [0,\infty) \times \BBR^n \to \BBR$ given by 
\begin{align}
    \ubar{h}_{5}(t,x) &\triangleq \| sy(\psi,\gamma) - \hat{w}(t) \|^2 - \ubar{v}^2 \notag \\
    & \qquad - 2\bar{w}\sqrt{\varepsilon + \| sy(\psi,\gamma) - \hat{w}(t) \|^2}, \label{eq:airspeed_cons_bound_1} \\
    \ubar{h}_{6}(t,x) &\triangleq \bar{v}^2 - \| sy(\psi,\gamma) - \hat{w}(t) \|^2 \notag \\
     & \qquad - 2\bar{w}\sqrt{\varepsilon + \| sy(\psi,\gamma) - \hat{w}(t) \|^2} - \bar{w}^2. \label{eq:airspeed_cons_bound_2}
\end{align}
The next result shows that the computable functions \cref{eq:airspeed_cons_bound_1,eq:airspeed_cons_bound_2} are lower bounds for \cref{eq:airspeed_cons_1,eq:airspeed_cons_2}. 
The proof is in \cref{app:prop_proofs}.

\begin{proposition}\rm 
\label{prop:robust_constraint}
Let $(t,x) \in [0,\infty) \times \BBR^{n}$. 
Then, $\ubar{h}_5(t,x) \leq h_5(t,x)$ and $\ubar{h}_6(t,x) \leq h_6(t,x)$.
\end{proposition}

Next, for all $t \ge 0$, define
\begin{gather}
\ubar{\SP}(t) \triangleq \{ x \in \SP(t) \colon \ubar{h}_5(t,x) \geq 0,\ubar{h}_6(t,x) \geq 0\},  \label{eq:P_bar_set}\\
 \ubar{\SC}(t) \triangleq \SSS(t) \cap \ubar{\SP}(t),\label{eq:C_bar_set}
\end{gather}
and it follows from \cref{prop:robust_constraint} that for all $t \ge 0$, $\ubar{\SP}(t) \subset \SP(t)$ and $\ubar{\SC}(t) \subset \SC(t)$.
%
%
%
%
%
The next result shows that forward invariance of $\ubar{\SP}(t)$ ensures forward invariance of $\ubar{\SC}(t)$.
The proof is in \cref{app:prop_proofs}.

\begin{proposition}\label{prop:forward_invariance}
Let $x_0 \in \ubar{\SC}(0)$, and let $t_\rmm \in (0,\infty]$.
Assume for all $t \in [0,t_\rmm)$, $x(t) \in \ubar{\SP}(t)$. 
Then, for all $t \in [0,t_\rmm)$, $x(t) \in \ubar{\SC}(t)$.
\end{proposition}

\cref{prop:forward_invariance} shows that forward invariance of $\ubar{\SP}(t)$ implies that for all $t \in [0,\infty)$, $x(t) \in \ubar{\SC}(t) \subset \SC(t) \subset \SSS(t)$.
Hence, we construct a candidate CBF whose zero-superlevel set is a subset of $\ubar{\SP}(t)$. 
Let $\rho > 0$, and consider $h \colon [0,\infty) \times\BBR^{n} \to \BBR$ defined by
\begin{align}
    h(t,x) &\triangleq \mathrm{softmin}_\rho \Big( h_1(x), \ldots ,h_4(x) , \ubar{h}_5(t,x), \ubar{h}_6(t,x) , \nn\\
    &\qquad h_7(t,x), \ldots, h_{r+6}(t,x) \Big ),\label{eq:h}
\end{align}
where 
\begin{equation*}
\mathrm{softmin}_\rho ( z_1, \ldots, z_N) \triangleq - \frac{1}{\rho} \mathrm{log} \sum_{i=1}^N e^{-\rho z_i}
\end{equation*}
is the log-sum-exponential \textit{soft minimum}. 
The zero-superlevel set of $h$ is
\begin{equation}\label{eq:h_set}
    \mathcal{H}(t) \triangleq \{x\in\BBR^{n}:h(t,x)\geq0\}.
\end{equation}
It follows from \cite[Proposition 2]{safari_time-varying_2025} that for fall $t \in [0,\infty)$, $\SH(t) \subset \ubar{\SP}(t)$ and as $\rho \to \infty$, $\SH(t) \to \ubar{\SP}(t)$.
The soft minimum is a continuously differentiable lower bound on the minimum. The worst case error bound of the soft minimum with respect to the minimum is $(\log N)/\rho$ \cite{rabiee_soft-minimum_2025}. 
Thus, selecting $\rho$ is a trade off between conservativeness with respect to $N$ and the magnitude of $\frac{\partial h(t,x)}{\partial x}$ at points where the minimum is not differentiable; for more details see \cite{rabiee_soft-minimum_2025}.

Next, consider the relaxed CBF constraint function $\nu \colon [0,\infty) \times \BBR^n \times \BBR^3 \times \BBR \to \BBR$ given by 
\begin{align}
    \nu(t,x,\bar u,\bar \kappa) &\triangleq \frac{\partial h(t,x)}{\partial t} + \frac{ \partial h(t,x)}{\partial x} \left [ f(x) + B \bar u \right ] \nn\\
    &\qquad + \alpha \big( h(t,x) \big ) + \bar \kappa h(t,x), \label{eq:nu}
\end{align}
where $\bar u \in \BBR^3$ is the control variable, $\bar \kappa \in \BBR$ is a slack variable, and $\alpha \colon \BBR \to \BBR$ is locally Lipschitz and nondecreasing such that $\alpha(0)=0$.
The CBF constraint $\nu(t,x,\bar u,\bar \kappa) \ge 0$ describes control and slack variables $\kappa \ge 0$ that make $\SH(t) \subset \ubar{\SP}(t)$ forward invariant, which, in turn, make $\ubar{\SC}(t) \subset \SSS(t)$ forward invariant (from \cref{prop:forward_invariance}). 
The slack term $\bar \kappa h(t,x)$ in \eqref{eq:nu} ensures that $\nu(t,x,\bar u,\bar \kappa) \ge 0$ is feasible for all $(t,x) \in [0,\infty) \times \BBR^n$ such that $x \not \in \mbox{bd } \SH(t)$.
The following result provides a sufficient condition such that $\nu(t,x,\bar u,\bar \kappa) \ge 0$ is feasible for all $(t,x) \in [0,\infty) \times \BBR^n$.  
The proof is in \cref{app:prop_proofs}.

\begin{proposition}\rm \label{prop:omega_feasibility}
Assume that for all $(t,x) \in [0,\infty) \times \in \mbox{bd } \SH(t)$, $\frac{ \partial h(t,x)}{\partial x} B \ne 0$. 
Then, for all $(t,x) \in [0,\infty) \times \BBR^n$, there exist $(\bar u, \bar \kappa) \in \BBR^3 \times \BBR$ such that $\nu(t,x,\bar u,\bar \kappa) \geq 0$.
\end{proposition}

For the remainder of this article, we make the following feasibility assumption. 
\begin{enumA}
    \item For all $(t,x) \in [0,\infty) \times \in \mbox{bd } \SH(t)$, $\frac{ \partial h(t,x)}{\partial x} B \ne 0$. \label{assump:Lgh}
\end{enumA}







\section{Receding-Horizon Optimal Control with State Constraints}

This section presents a receding-horizon control using the C-ADP approach and associated closed-form solution from \cite{gutierrez_receding-horizon_2026}.
First, we define 
\begin{gather*}
    a(t, x) \triangleq \frac{\partial h(t,x)}{\partial t} + \frac{ \partial h(t,x)}{\partial x} f(x) + \alpha\big(h(t,x)\big), \\
    b(t, x)^\top \triangleq \begin{bmatrix}
        \frac{ \partial h(t,x)}{\partial x} B & 
        h(t,x)
    \end{bmatrix},
\end{gather*}
and it follows that $\nu(t,x,\bar u, \bar \kappa) = a(t,x) + b(t,x)^\top \bar \upsilon$, where  
    $\bar \upsilon \triangleq \left [ \begin{smallmatrix}
        \bar u \\
        \bar \kappa
    \end{smallmatrix} \right ]$.
This constraint is an affine function of $\bar u$ and $\bar \kappa$, whereas the cost \cref{eq:cost_function} is quadratic in $\bar{u}$ and does not depend on $\bar \kappa$.  
Thus, let $r_{\kappa}>0$, and consider the regularized cost
\begin{equation}\label{eq:regularized_cost}
    \bar J(t,x_0,\bar{\upsilon}(\cdot))
    \triangleq J(t,x(\cdot),\bar{u}(\cdot)) + \frac{r_{\kappa}}{2} \int_t^{t+T} \bar{\kappa}(\tau)^2 \, \mathrm{d}\tau.
\end{equation}
We use a receding-horizon C-ADP approach to construct $\bar \upsilon$ such that for all $t \ge 0$, $\bar J(t,x_0,\bar{\upsilon}(\cdot))$ is approximately minimized subject to the constraint $a(\tau,x(\tau)) + b(\tau,x(\tau))^\top \bar \upsilon(\tau) \ge 0$ for all $\tau \in [t,t+T]$.

\subsection{C-ADP Solution}

The C-ADP method uses time discretization of the planning horizon $T$. 
Let $T_{\rm p}>0$ be the planning time step such that $N=T/T_{\rm p}$ is an integer, and define $\SN \triangleq \{0, \ldots, N-1\}$.
Consider the discretization of \cref{eq:system_dynamics,eq:system_dynamics.2} given by 
\begin{equation}\label{eq:cadp_dynamics}
    x_{i+1} = F(x_i) + G(x_i) \bar \upsilon_i,
\end{equation}
where the discretization time is $T_{\rm p}$, and $F$ and $G$ are assumed to be continuously differentiable on $\SD$. 
Note that the right-most column of $G(x_i)$ is zero because \cref{eq:system_dynamics,eq:system_dynamics.2} does not depend on $\bar \kappa$, which is the bottom element of $\bar \upsilon$.
The discrete model \eqref{eq:cadp_dynamics} can be obtained from a variety of discretization approaches. 
For example, a simple forward-Euler discretization yields $F(x) = x + T_{\rm p}f(x)$ and $G(x) = \begin{bmatrix} T_{\rm p}B & 0_{n \times 1}\end{bmatrix}$. 
However, higher-order methods can be used.

We note that the discretization of $\bar \upsilon$ over the planning horizon $T$ and the associated discrete model \eqref{eq:cadp_dynamics} are used only to optimize the cost $\bar J$. 
The C-ADP approach yields closed-form feedback functions for each step of the planning horizon $T$. 
Then, these closed-form feedback functions are updated in real time using a receding horizon.
Notably, the CBF constraint is satisfied continuously for all real time $t$ and planning time $\tau$.

The $N$ control functions $\bar \upsilon_0,\ldots,\bar \upsilon_N$ are updated with period $T_{\rm s}\in(0,T_{\rm p}]$, which implies that the receding horizon update can be faster than the sample time $T_\rmp$ used for discretization of the plan.  
For all $(i,k) \in \SN \times \BBN$, define the prediction times $t_{i,k} \triangleq iT_{\rm p}+kT_{\rm s}$. 
Then, for all $k \in \BBN$, the cost \cref{eq:regularized_cost} at time $kT_\rms$ is discreteized as
\begin{equation}\label{eq:discrete_cost}
    \bar J_k(x_0,\bar \upsilon_0,\ldots,\bar \upsilon_{N-1}) \triangleq \ell_{N,k}(x_{N}) + \sum_{i=0}^{N-1} \ell_{i,k}(x_i, \bar \upsilon_i),
\end{equation}
where $x_1,\ldots,x_N$ satisfy \eqref{eq:cadp_dynamics}, and 
\begin{gather}
\ell_{i,k}(x, \bar \upsilon) \triangleq \frac{1}{2}x^{\top}Q_{i,k}x +\Gamma_{i,k}^{\top}x \! + \!\frac{1}{2}\bar  \upsilon^{\top}R_{i,k} \bar \upsilon +\Omega_{i,k}^{\top} \bar \upsilon ,\label{eq:stage_cost}\\
\ell_{N,k}(x) \triangleq \frac{1}{2}x^{\top}Q_{N,k}x +\Gamma_{N,k}^{\top}x, \label{eq:terminal_cost}
\end{gather}
where
\begin{align}
    Q_{i,k}&\triangleq Q(t_{i,k}), &
    \Gamma_{i,k}&\triangleq \Gamma(t_{i,k}), \label{eq:discrete_state_weights}\\
    R_{i,k}&\triangleq 
    \begin{bmatrix}
        R(t_{i,k})&0_{3\times1}\\
        0_{1\times3}&r_\kappa
    \end{bmatrix}, &
    \Omega_{i,k}&\triangleq 
    \begin{bmatrix}
        \Omega(t_{i,k})\\0
    \end{bmatrix}. \label{eq:discrete_input_weights}
\end{align}

Finally, for all $(i,k) \in \SN \times \BBN$, consider the \textit{C-ADP control functions} $\upsilon^\star_{i,k} \colon [0,\infty) \times \SD \to \BBR$ defined by
\begin{equation}\label{eq:cadp_policy}
\upsilon^\star_{i,k}(t,x) = \zeta_{i,k}(x) + \lambda_{i,k}(t,x) W_{i,k}(x) b(t, x),
\end{equation}
where
\begin{align}
\lambda_{i,k}(t,x) &\triangleq \max \left \{0,\frac{ -a(t, x)-b(t,x)^\top\zeta_{i,k}(x)}{b(t, x)^\top W_{i,k}(x) b(t, x) } \right \},\label{eq:cadp_lambda}\\
\zeta_{i,k}(x) &\triangleq -W_{i,k}(x) [ G(x)^\top\bigl(P_{i+1,k}F(x)+T_{i+1,k}\bigr) + \Omega_{i,k}  ], \label{eq:zeta_ik} \\
W_{i,k}(x) &\triangleq \left(R_{i,k}+G(x)^\top P_{i+1,k}G(x)\right)^{-1}, \label{eq:W_ik}
\end{align}
where
\begin{align}
P_{i,k} &\triangleq Q_{i,k}+K_{i,k}^\top R_{i,k}K_{i,k} + \tilde{A}_{i,k}^\top P_{i+1,k}\tilde{A}_{i,k}, \label{eq:cadp_P_recursion} \\
T_{i,k} &\triangleq  \tilde{A}_{i,k}^\top P_{i+1,k} \bigl(F(\bar{x}_{i,k}) + G(x)\upsilon^\star_{i,k}(t_{i,k},\bar{x}_{i,k}) - \tilde{A}_{i,k}\bar{x}_{i,k} \bigr) \notag\\
 & \quad  + K_{i,k}^\top \bigl(R_{i,k}(\upsilon^\star_{i,k}(t_{i,k},\bar{x}_{i,k}) -K_{i,k}\bar{x}_{i,k})+\Omega_{i,k}\bigr) \notag \\
 & \quad +\tilde{A}_{i,k}^\top T_{i+1,k} + \Gamma_{i,k},
\label{eq:cadp_T_recursion}
\end{align}
where $P_{N,k}=Q_{N,k}$ and $T_{N,k}=\Gamma_{N,k}$, and 
\begin{gather}
K_{i,k} \triangleq \left. \frac{\partial \tilde{\upsilon}^\star_{i,k}(x)}{\partial x} \right|_{x=\bar{x}_{i,k}}, \quad
\tilde{A}_{i,k} \triangleq \left. \frac{\partial \tilde{F}^\star_{i,k}(x)}{\partial x} \right|_{x=\bar{x}_{i,k}}, \label{eq:cadp_gains}\\
\tilde{\upsilon}^\star_{i,k}(x) \triangleq \zeta_{i,k}(x) + \tilde{\lambda}_{i,k}(x) W_{i,k}(x)b(t_{i,k}, x), \label{eq:cadp_smooth_policy}\\
\tilde{F}^\star_{i,k}(x) \triangleq F(x) + G(x) \tilde{\upsilon}^\star_{i,k}(x),
\label{eq:cadp_smooth_dynamics}\\
\tilde{\lambda}_{i,k}(x) \triangleq \operatorname{softplus}_{\eta} \left (\frac{-a(t_{i,k}, x) -b(t_{i,k}, x)^\top\zeta_{i,k}(x)}{b(t_{i,k}, x)^\top W_{i,k}(x) b(t_{i,k}, x)} \right ),
\end{gather}
where $\operatorname{softplus}_{\eta}(z) \triangleq \frac{1}{\eta}\log\bigl(1+e^{\eta z}\bigr)$, $\eta>0$, and $\{\bar{x}_{i,k}\}_{i=0}^{N}$ is a nominal prediction trajectory used for the approximations \cref{eq:cadp_gains,eq:cadp_smooth_policy,eq:cadp_smooth_dynamics}.

For each $k\in \BBN$, the functions $\upsilon^{\star}_{N-1,k}, \ldots, \upsilon^{\star}_{0,k}$ are solved using a backward recursion, where \cref{eq:cadp_policy,eq:cadp_lambda,eq:zeta_ik,eq:W_ik,eq:cadp_P_recursion,eq:cadp_T_recursion} are solved from $i = N-1$ to $i = 0$ with the terminal conditions $P_{N,k}=Q_{N,k}$ and $T_{N,k}=\Gamma_{N,k}$.
We note that the C-ADP solution \cref{eq:cadp_policy,eq:cadp_lambda,eq:zeta_ik,eq:W_ik,eq:cadp_P_recursion,eq:cadp_T_recursion} simplifies to the classic finite-horizon linear quadratic regulator in the case of linear dynamics and no active constraints. 
See \cite[Remark 1]{gutierrez_receding-horizon_2026} for more details.

For each $k \in \BBN$, the \textit{C-ADP control plan over the horizon $[k T_{\rm s}, k T_{\rm s} + T]$} is $\upsilon_{k}^{\star} \colon [k T_{\rm s}, k T_{\rm s} + T] \times \SD \to \BBR^4$, where for all $i \in \SN$ and all $\tau\in [t_{i,k}, t_{i+1,k})$, 
\begin{equation} \label{eq:CADP_plan}
    \upsilon_{k}^{\star}(\tau, x) \triangleq  \upsilon^{\star}_{i,k}(\tau,x). 
\end{equation}
For $k=0$, a user-specified initial nominal prediction trajectory $\{\bar{x}_{i,0}\}_{i=0}^N$ is used to compute $P_{i,0}$ and $T_{i,0}$. 
For $k\in\BBN \setminus \{ 0 \}$, the C-ADP control plan $\upsilon_{k-1}^\star(\tau, x)$ from the $k-1$ backward pass is used in a forward pass to generate the prediction trajectory $\{\bar{x}_{i,k}\}_{i=0}^N$. 
More specifically, $\{\bar{x}_{i,k}\}_{i=0}^N$ is obtained from a sampling of the solution to the system dynamics \cref{eq:system_dynamics,eq:system_dynamics.2} with control $u(t) = [ \begin{smallmatrix}I_{3\times3}&0_{3\times1}\end{smallmatrix} ]\upsilon_{k-1}^\star(t-k T_\rms, x)$ with initial condition $x_{0} = x(t_{0,k})$.

\subsection{Properties of the C-ADP Control Plan}

The next result shows that for each $k \in \BBN$, the closed-form C-ADP control plan \eqref{eq:CADP_plan} satisfies the constraint $a(\tau,x) + b(\tau,x)^\top \upsilon_k^\star(\tau,x) \ge 0$ over the planning horizon $\tau \in [kT_\rms , k T_\rms + T ]$.

\begin{theorem}\label{thm:planning_safety}
Assume that \ref{assump:Lgh} is satisfied. 
Then, for all $k \in \BBN$ and all $(\tau,x) \in [k T_{\rm s},k T_{\rm s}+T] \times \SD$,
\begin{equation*}
        a(\tau,x) + b(\tau,x)^\top \upsilon_{k}^{\star}(\tau,x) \ge 0.
    \end{equation*}
\end{theorem}

\begin{proof}
Let $k \in \BBN$ and $i \in \SN$.
Since $R_{i,k}$ is positive definite, it follows that $W_{i,k}$ is positive definite. 
Thus, \ref{assump:Lgh} implies that for all $(\tau,x) \in [k T_{\rm s},k T_{\rm s}+T] \times \BBR^n$, $b(\tau,x)^\top W_{i,k}(x)b(\tau,x)>0$.
Next, substituting \cref{eq:cadp_policy,eq:cadp_lambda,eq:zeta_ik} into \cref{eq:nu} yields $a(\tau,x)+b(\tau,x)^\top\upsilon_{i,k}^\star(\tau,x) =\max \{0,a(\tau,x)+b(\tau,x)^\top\zeta_{i,k}^\star(x)\}\geq 0$.
\end{proof}

\Cref{thm:planning_safety} shows that for each $k \in \BBN$, the C-ADP control plan \eqref{eq:CADP_plan} satisfies the CBF safety constraint over the associated planning horizon.



Next, we examine the optimality properties of the C-ADP control functions \cref{eq:cadp_policy}. 
For all $(i,k) \in \SN \times \BBN$, the C-ADP control functions \cref{eq:cadp_policy,eq:cadp_lambda,eq:zeta_ik,eq:W_ik,eq:cadp_P_recursion,eq:cadp_T_recursion} are obtained from a quadratic approximation of the cost-to-go along the nominal trajectory at step $i$ of the backward recursion.
For all $(i,k) \in \SN \times \BBN$, this approximate cost at step $i$ is given by 
\begin{align}
\tilde{J}_{i,k}(x, \bar \upsilon ) &\triangleq \ell_{i,k}(x, \bar{\upsilon}) \nn \\
    & \quad + \frac{1}{2} \left [F(x) + G(x) \bar \upsilon \right ]^\top P_{i+1,k} \left [F(x) + G(x) \bar \upsilon \right ] \nn \\
    & \quad + T_{i+1,k}^\top \left [F(x) + G(x) \bar \upsilon \right ], \label{eq:approximate_cost}
\end{align}
which is quadratic and strictly convex in $\bar{\upsilon}$. 
See \cite{gutierrez_receding-horizon_2026} for additional details. 
The next result shows that for all $(i,k) \in \SN \times \BBN$, the C-ADP control function \cref{eq:cadp_policy} is the unique global minimizer of \eqref{eq:approximate_cost} subject to the CBF constraint.
The proof is similar to that of \cite[Theorem 1]{safari_time-varying_2025} and is omitted for space.

\begin{theorem}\label{them:optimality}
Let $(i,k) \in \SN \times \BBN$ and $(\tau,x) \in [k T_{\rm s}, k T_{\rm s} + T] \times \SD$, and assume that \ref{assump:Lgh} is satisfied. 
Let $\upsilon \in \BBR^4$ be such that $a(\tau,x) + b(\tau,x)^\top \upsilon \geq 0$ and $\upsilon \neq \upsilon_{i,k}^\star (\tau, x)$. 
Then, $\tilde{J}_{i,k}(x, \upsilon) > \tilde{J}_{i,k}(x, \upsilon_{i,k}^\star(\tau,x))$.
\end{theorem}

\subsection{Receding-Horizon C-ADP Control}

The real-time receding-horizon C-ADP control is constructed using the first function $\upsilon^{\star}_{0,k}$ each time the C-ADP control plan is updated. 
Specifically, for all $k \in \BBN$ and all $t \in [kT_{\rm s}, (k+1)T_{\rm s})$, the C-ADP control $u_{\star}: [0,\infty) \times \SD \to \BBR^3$ is given by
\begin{equation}\label{eq:optimal_safe_control}
u_{\star}(t,x) \triangleq \begin{bmatrix} I_{3 \times 3} & 0_{3 \times 1} \end{bmatrix} \upsilon^{\star}_{0,k}(t,x),
\end{equation}
and the C-ADP slack parameter is  $\kappa_{\star}: [0,\infty) \times \SD \to \BBR$ is
\begin{equation}
\kappa_{\star}(t,x) \triangleq \begin{bmatrix} 0_{1 \times 3} & 1 \end{bmatrix} \upsilon^{\star}_{0,k}(t,x).
\end{equation}
The receding-horizon C-ADP control is summarized in \cref{alg:C-ADP}.
Note that $u_{\star}$ and $\kappa_{\star}$ are piecewise continuous in $t$ and continuous in $x$ on $\SD$. 
In addition, $u_{\star}$ and $\kappa_{\star}$ are continuous in $t$ on each interval $[kT_{\rm s}, (k+1)T_{\rm s})$ but are generally discontinuous at $kT_{\rm s}$, which corresponds to the $\upsilon^{\star}_{0,k}$ update. 
Moreover, if $\frac{ \partial h(t,x)}{\partial x}$ is locally Lipschitz in $x$ on $\SD$, then \cref{eq:cadp_policy,eq:cadp_lambda,eq:W_ik,eq:zeta_ik} imply that $u_{\star}$ is locally Lipschitz in $x$ on $\SD$.

\begin{algorithm} 
  \caption{Receding-Horizon C-ADP Control} \label{alg:C-ADP}
  \KwInput{$\{\bar{x}_{i,0}\}_{i\in\mathcal{N}}, Q, \Gamma, R,
    \Omega, r_\kappa, T_{\rm p}, T, a, b, F, G$}
  \KwOutput{$\{v_{i,k}^\star(t,x)\}_{i\in\mathcal{N}}$}

  \For{$k \gets 0,1,2,\ldots$}{
    \If{$k>0$}{
      \Comment{Forward pass}
      $\bar{x}_{0,k} \gets x(t_{0,k})$\;
      \For{$i \gets 0,\ldots,N-1$}{
        Solve \cref{eq:system_dynamics,eq:system_dynamics.2} on $[iT_{\rm p},(i+1)T_{\rm p}]$ with $x_0=\bar{x}_{i,k}$ and $u= [ \begin{smallmatrix}I_{3\times3}&0_{3\times1}\end{smallmatrix} ] \upsilon^\star_{i,k-1}$\;
        $\bar{x}_{i+1,k} \gets x((i+1)T_{\rm p})$\;
      }
    }
    \Comment{Backward pass}
    $\{\upsilon^{\star}_{i,k}\}_{i\in\mathcal{N}} \gets \text{\cref{eq:cadp_policy,eq:cadp_lambda,eq:zeta_ik,eq:W_ik,eq:cadp_P_recursion,eq:cadp_T_recursion}}$\;
    \Comment{Extract control function}
    $u_*(t,x) \gets \text{\cref{eq:optimal_safe_control}}$\;
  }
\end{algorithm}

The following result is the main result on constraint satisfaction for the receding-horizon C-ADP control.
The proof is similar to \cite[Theorem 2]{safari_time-varying_2025} and is omitted for space.

\begin{theorem}\label{eq:theorem_main}
Consider \cref{eq:system_dynamics,eq:system_dynamics.2}, where   \ref{assump:Lgh} is satisfied. 
Let $u = u_{\star}$, where $u_{\star}$ is given by \cref{eq:optimal_safe_control}. 
Assume $\frac{ \partial h(t,x)}{\partial x}$ is locally Lipschitz in $x$ on $\SD$.
Then, for all $x_0 \in \ubar{\mathcal{C}}(0)$, there exists $t_{\rm m} \in (0,\infty]$ such that \cref{eq:system_dynamics,eq:system_dynamics.2} with $u = u_{\star}$ has a unique solution on $[0, t_{\rm m})$.
Furthermore, For all $t \in [0, t_{\rm m})$, $x(t) \in \ubar{\mathcal{C}}(t) \subset \SSS(t)$.
\end{theorem}

\Cref{eq:theorem_main} shows that the C-ADP control ensures all constraints are satisfied in real time $t$, and \Cref{thm:planning_safety} shows that for each update $k \in \BBN$, the C-ADP control plan \eqref{eq:CADP_plan} satisfies constraints in planing time $\tau$ over the planning horizon $T$.  
In addition, \Cref{them:optimality} describes the optimality property of the C-ADP control functions \eqref{eq:cadp_policy}.

\section{Numerical Simulations}

Consider an airspace with an admissible altitude of $10 \leq e_3^\top q \leq 300\,\mathrm{m}$, an operational geofence of $-200 \leq e_1^\top q \leq 200\,\mathrm{m}$ and $-150 \leq e_2^\top q \leq150\,\mathrm{m}$, and $r_{\rm obs}$ static spherical obstacles. The airspace is the intersection of the zero-superlevel sets of
\begin{align*}
b_1(q) &\triangleq e_3^\top q-10, &
b_2(q) &\triangleq 300-e_3^\top q,\\
b_3(q) &\triangleq e_1^\top q+200, &
b_4(q) &\triangleq 200-e_1^\top q,\\
b_5(q) &\triangleq e_2^\top q+150, &
b_6(q) &\triangleq 150-e_2^\top q,\\
b_{j+6}(q) &\triangleq \|q-q_{\mathrm{obs},j}\|^2-\delta_{\mathrm{obs},j}^2,
&& j=1,\ldots,r_{\rm obs},
\end{align*}
where $q_{\mathrm{obs},j}\in\mathbb{R}^3$ and $\delta_{\mathrm{obs},j} > 0$ is the position and radius of the $j$-th obstacle respectively. 

The wind velocity is $w \in \mathcal{N}(\mu,\Sigma)$, where $\mathcal{N}$ is Gaussian white noise with mean $\mu = \begin{bmatrix} 4.34 & 1.08 & 0 \end{bmatrix}^{\top} \, \mathrm{m/s}$ and covariance $\Sigma = 0.04 \, I_3 \, \mathrm{m/s}$. The measured wind velocity is $\hat{w} = \mu$ and the measurement error bound is $\bar{w} = 0.6 \, \mathrm{m/s}$.

Let $q_{\rm g} \in \BBR^3$ be the goal position. The objective is to navigate to the goal position, while remaining inside the airspace and satisfying state constraints. We say the aircraft \textit{reaches the target} at time $t_{\rm f}$ if the position satisfies $\|q(t_{\rm f}) -q_{\rm g}\| \leq 3.0 \, \rm{m}$. We say the aircraft \textit{does not reach the target} if  $t = T_{\rm f} = 50 \, \rm{s}$ and $\|q(T_{\rm f}) -q_{\rm g}\| > 3.0 \, \rm{m}$.

We compare the receding-horizon C-ADP control \cref{eq:optimal_safe_control} with two other control approaches. In the three approaches, the airspeed barrier functions \cref{eq:airspeed_cons_bound_1,eq:airspeed_cons_bound_2} use $\ubar{v} = 2 \, \rm m/\rm s$, $\bar{v} =20 \, \rm m/\rm s$ and $\varepsilon = 10^{-6}$. The barrier functions \cref{eq:flight_path_constraint,eq:roll_constraint} use $\ubar{\gamma} = - \frac{2 \pi}{9} \, \rm{rad}$, $\bar{\gamma} = \frac{2 \pi}{9} \, \rm{rad}$, $\ubar{\phi} = -\frac{5 \pi}{18} \, \rm{rad}$ and $\bar{\phi} = \frac{5 \pi}{18} \, \rm{rad}$. For the CBF constraint, we use $\beta_i(h) = 2 \,h$ in \cref{eq:h_11}, $\rho = 40$ in \cref{eq:h} and $\alpha(h) = h$ in \cref{eq:nu}. For sample-data implementation, all controllers are executed at $100 \, \rm{Hz}$ with a zero-order hold.

\textbf{Myopic Minimum Intervention (MMI)}: 
A desired control $u_{\rm d}: \BBR^n \to \BBR^3$ that achieves the objective without satisfying state constraints is designed. 
Define the desired pointing direction $y_{\rm d}(x) \triangleq -\frac{q - q_{\rm g}}{\|q - q_{\rm g}\|}$. 
Then, the desired heading and flight-path angle are
\begin{align*}
    \psi_{\rm d}(x) &\triangleq \operatorname{atan2}(e_2^\top y_{\rm d}(x), e_1^\top y_{\rm d}(x)), \\
    \gamma_{\rm d}(x) &\triangleq -\arcsin(e_3^\top y_{\rm d}(x)). \label{eq:desired_gamma}
\end{align*}
To obtain the largest speed that satisfies
\cref{eq:airspeed_cons_bound_2}, let $\delta > 0$ and define $c \triangleq \bar v^2+\bar w^2 -\sqrt{4\bar v^2\bar w^2+\varepsilon} - \delta$. Then, consider
\begin{equation*}
s_{\rm d}(t,x) \triangleq y_{\rm d}(x)^\top\hat w + \sqrt{ c - \left\| \left(I-y_{\rm d}(x)y_{\rm d}(x)^\top\right)\hat{w}(t) \right\|^2}.
\end{equation*}

Consider the wrapping function $\varphi: \BBR \to (-\pi, \pi]$ defined as $\varphi(z) \triangleq \operatorname{atan2}(\sin z, \cos z)$. Next, let $a_{\psi} > 0$ and consider the desired roll angle
\begin{multline*}
    \phi_{\rm d}(x) \triangleq \tan^{-1} \Big( \frac{s \cos\gamma}{g_{\rm E}}
    \big[ \dot{\psi}_{\rm d}(x) - a_{\psi} \varphi(\psi - \psi_{\rm d}(x) ) \big] \Big).
\end{multline*}
Next, let $\sigma > 0$ and consider the desired control $u_{\rm d} \colon [0,\infty) \times \BBR^{n} \to \BBR^{3}$ defined by
\begin{equation} \label{eq:u_d}
    u_{\rm d}(t,x) \triangleq \begin{bmatrix} \dot{s}_{\rm d}(t,x) \\ \dot{\gamma}_{\rm d}(x) \\ \dot{\phi}_{\rm d}(x) \end{bmatrix} - \sigma \left(\begin{bmatrix} s \\ \gamma \\ \phi \end{bmatrix} - \begin{bmatrix} s_{\rm d}(t,x) \\ \gamma_{\rm d}(x) \\ \phi_{\rm d}(x) \end{bmatrix} \right).
\end{equation}

We use the desired control parameters $\delta = 10^{-3}$, $a_{\psi} = 1 $, $\sigma = 3$. Then, the implemented control is computed from a minimum-intervention QP with the CBF-constraint, namely
\begin{subequations}\label{eq:minimum_intervention}
\begin{align}
    &\big( u(t,x), \kappa(t,x) \big) = \argmin_{\bar{u}, \bar{\kappa}} \| \bar{u} - u_{d}(t,x) \|^2 + r_{\kappa} \bar{\kappa}^2 ,\\
    & \qquad \text{subject to} \quad a(t, x) + b(t, x)^\top \left[ \begin{smallmatrix}
        \bar u \\
        \bar \kappa
    \end{smallmatrix} \right] \geq 0.
\end{align}
\end{subequations}
with $r_{\kappa} = 10^8$.

\textbf{iLQR with Minimum Intervention (iLQR+MI)}: We compute a desired control $u_{\rm d}$ using iLQR, which is equivalent to minimizing the cost \cref{eq:cost_function} without the CBF constraint. 
Define $x_{\rm g}\triangleq\begin{bmatrix}q_{\rm g}^{\top}&0&0&0&0\end{bmatrix}^{\top}$.
Then, the cost function \cref{eq:regularized_cost}, with $T = 10.0 \, \rm{s}$ is given by $Q=\operatorname{blkdiag}(I_{3 \times 3},0,0,0,0)$, $\Gamma=-x_{\rm g}$, $R=I_{3 \times 3}$ and $\Omega = 0$. The update period and planning time step are $T_{\rm s} = T_{\rm p} = 0.05 \, \rm{s}$.
Since iLQR does not guarantee constraint satisfaction, the implemented control is computed from a minimum-intervention QP with the CBF-constraint as in \cref{eq:minimum_intervention}.

\textbf{Constrained-Approximate Dynamics Programming (C-ADP) Control}:
The control \cref{eq:optimal_safe_control} with $T = 10.0 \, \rm{s}$ is implemented with $Q=\operatorname{blkdiag}(I_{3 \times 3},0,0,0,0)$, $\Gamma=-x_{\rm g}$, $R=I_{3 \times 3}$, $\Omega = 0$ and $r_{\kappa} = 10^8$. The update period and planning time step are $T_{\rm s} = T_{\rm p} = 0.05 \, \rm{s}$.

\subsection{Single Maneuver Analysis}

Consider the initial position $q_{\rm 0} = \begin{bmatrix} -90 & -20 & 110 \end{bmatrix}^{\top}\,{\rm m} $ and initial speed $s_{\rm 0}=10\,{\rm m/s}$. The initial heading, flight-path, and roll angles are zero. The goal position is $ q_{\rm g} = \begin{bmatrix} 150 & 60 & 150 \end{bmatrix}^{\top}\,{\rm m}$.

\begin{figure}
    \centering
    \includegraphics[width=0.99\columnwidth]{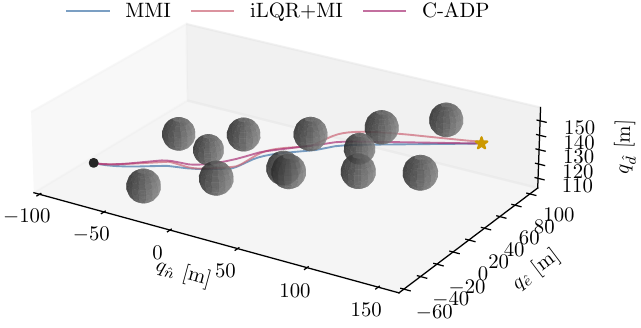}
    \caption{Position trajectories for the executed maneuvers. The black circle and gold star are the initial and goal positions, respectively.}
    \label{fig:traj_position}
\end{figure}

\begin{figure}
    \centering
    \includegraphics[width=0.9\columnwidth]{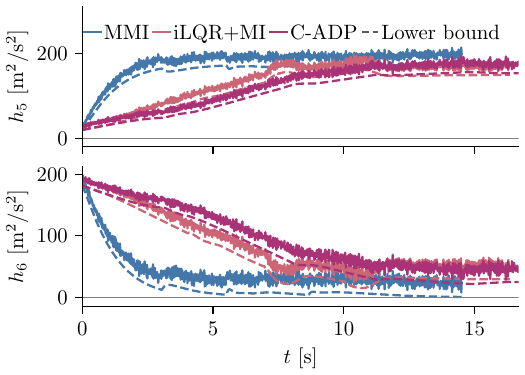}
    \caption{Airspeed barrier functions. Dashed curves show the computable lower bounds.}
    \label{fig:traj_cons}
\end{figure}

\begin{figure}
    \centering
    \includegraphics[width=0.9\columnwidth]{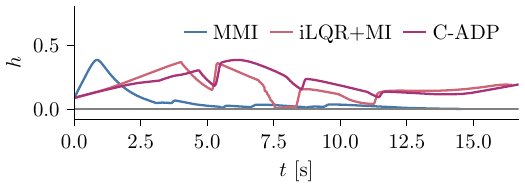}
    \caption{Composite CBF $h$ for the MMI, iLQR+MI and C-ADP controllers.}
    \label{fig:softmin}
\end{figure}

\Cref{fig:traj_position} compares the position trajectories. \Cref{fig:traj_cons} shows the airspeed barrier functions and their computable lower bounds. For all $t \in [0,t_{\rm f}]$, $\ubar{h}_5(t,x(t)) \leq h_5(t,x(t))$ and $\ubar{h}_6(t,x(t)) \leq h_6(t,x(t))$. Moreover, \Cref{fig:softmin} shows that the composite CBF $h$ remains nonnegative throughout the maneuver, which implies that $x(t) \in \mathcal{S}(t)$ for all $t \in [0,t_{\rm f}]$. Note that the C-ADP trajectory is always in the interior of $\SH(t)$ while both MMI and iLQR+MI have segments of time in which $x(t) \in \mbox{bd } \SH(t)$, which makes them prone to violating the state constraints. 

\begin{figure}
    \centering
    \includegraphics[width=0.99\columnwidth]{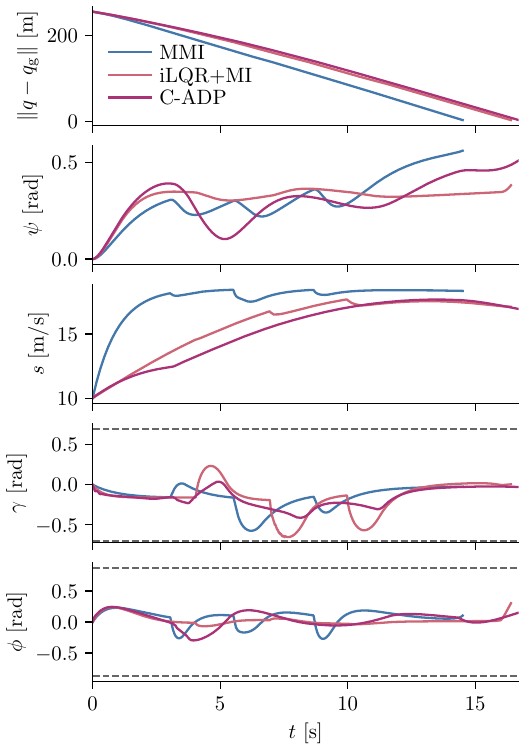}
    \caption{State histories for the MMI, iLQR+MI and C-ADP controllers. Dashed lines denote the flight-path and roll-angle limits.}
    \label{fig:traj_states}
\end{figure}

\begin{figure}
    \centering
    \includegraphics[width=0.99\columnwidth]{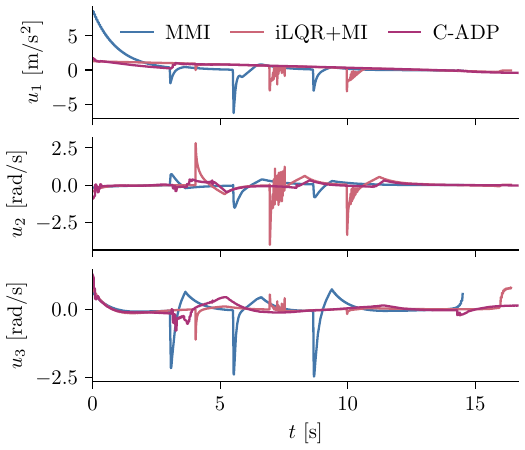}
    \caption{Control inputs for the MMI, iLQR+MI and C-ADP controllers.}
    \label{fig:traj_inputs}
\end{figure}

MMI completes the maneuver in the least amount of time, reaching the target at $t_{\rm f}=14.51,{\rm s}$. While iLQR+MI and C-ADP reach the target at $t_{\rm f}=16.40,{\rm s}$ and $t_{\rm f}=16.71,{\rm s}$, respectively. \Cref{fig:traj_states} shows the position error and state histories. MMI maintains a higher speed over much of the maneuver, which contributes to its shorter arrival time. However, this behavior is accompanied by larger control inputs, as shown in \Cref{fig:traj_inputs}. Both iLQR+MI and C-ADP optimize performance over a finite prediction horizon. The distinction is in how safety is incorporated: iLQR+MI enforces safety pointwise using the minimum-intervention CBF filter, whereas C-ADP incorporates the CBF constraint throughout the prediction horizon.

\subsection{Performance Study}

We compare the three controllers over the 16 maneuvers formed by four initial conditions and four goal positions. We use four performance metrics to compare the controllers: the time at which the aircraft reaches the target, the trajectory performance (TP), the control performance (CP), and the control intervention (CI), defined as $\mathrm{TP} \triangleq \int_0^{t_{\rm f}} \|q(\tau)-q_{\rm g}\|^2 \, \mathrm{d}\tau$, $\mathrm{CP} \triangleq \int_0^{t_{\rm f}} \|u(\tau)\|^2 \, \mathrm{d}\tau$ and $\mathrm{CI} \triangleq \max_{t \in [0,t_{\rm f}]} \|u(t)-u_{\rm d}(t)\|$.
%
%
For MMI and iLQR+MI, $u_{\rm d}$ is the desired control. For C-ADP, $u_{\rm d}$ is the unconstrained component of the C-ADP control obtained by setting $\lambda_{0,k}=0$ in \cref{eq:cadp_policy}.

\begin{table}[ht]
\caption{Controller performance over 16 maneuvers.}
\label{tab:controller_comparison}
\centering
\footnotesize
\begin{tabular*}{\columnwidth}{@{\extracolsep{\fill}}lccc@{}}
\hline
Metric & MMI & iLQR+MI & C-ADP \\
\hline
Targets reached
    & $\mathbf{16/16}$
    & $12/16$
    & $\mathbf{16/16}$ \\
$t_{\mathrm{f}}$ (s)
    & $14.9 \pm 0.95$
    & $22.7 \pm 5.8$
    & $\mathbf{14.7 \pm 0.98}$ \\
$\mathrm{TP}/10^{2}$
    & $3.61 \pm 0.61$
    & $4.25 \pm 0.61$
    & $\mathbf{3.51 \pm 0.59}$ \\
$\mathrm{CP}$
    & $37.93 \pm 7.79$
    & $\mathbf{17.77 \pm 9.86}$
    & $45.86 \pm 4.93$ \\
$\mathrm{CI}$
    & $4.65 \pm 4.2$
    & $5.68 \pm 4.6$
    & $\mathbf{3.72 \pm 1.1}$ \\
\hline
\end{tabular*}
\vspace{2pt}
\parbox{\columnwidth}{\scriptsize
Metric entries report the mean $\pm$ standard deviation computed over the completed maneuvers.}
\end{table}

\Cref{fig:trajectories} shows the resulting position trajectories, and \cref{tab:controller_comparison} summarizes the performance of the three controllers. Across the 16 start-goal pairs, MMI and C-ADP reached all targets, whereas iLQR+MI reached 12. Over the 12 maneuvers completed by all controllers, C-ADP had the lowest mean TP, arrival time, and CI, but the largest mean CP. In contrast, iLQR+MI had the lowest mean CP, but the largest mean TP, arrival time, and CI. MMI achieved performance between the two predictive controllers, with an arrival time comparable to C-ADP. Thus, C-ADP achieved the lowest mean trajectory cost and safety intervention while maintaining an arrival time comparable to MMI, at the expense of a larger control-performance metric.

\begin{figure}[ht]
    \centering
    \includegraphics[width=0.92\columnwidth]{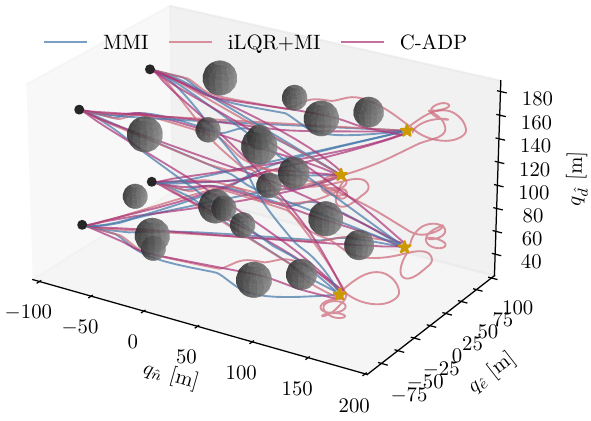}
    \caption{Position trajectories for the 16 maneuvers executed with the MMI, iLQR+MI and C-ADP controllers. Black circles are initial positions and gold stars are goal positions. Failed iLQR+MI maneuvers loiter around goal position.}
    \label{fig:trajectories}
\end{figure}

\appsection[]{}\label{app:prop_proofs}

\begin{proof}[Proof of \cref{prop:robust_constraint}]
Let $z \triangleq sy(\psi,\gamma)$ and define $\tilde{w}(t) \triangleq w(t)-\hat{w}(t)$.
First, we consider $\ubar{h}_5$. It follows from \cref{eq:airspeed_cons_1,eq:airspeed_cons_bound_1} that
\begin{align*}
    h_5(t,x)-\ubar{h}_5(t,x) &= \|z-w(t)\|^2-\|z-\hat{w}(t)\|^2 \\
    & \qquad + 2\bar{w}\sqrt{\varepsilon+\|z-\hat{w}(t)\|^2} \\
    &= -2 \big( z-\hat{w}(t) \big)^\top  \tilde{w}(t) + \|\tilde{w}(t)\|^2 \\
    & \qquad + 2\bar{w}\sqrt{\varepsilon+\|z-\hat{w}(t)\|^2} \\
    & \geq -2\|z-\hat{w}(t)\|\|\tilde{w}(t)\| + \|\tilde{w}(t)\|^2 \\
    & \qquad + 2\bar{w}\|z-\hat{w}(t)\| \\
    &=2\|z-\hat{w}(t)\| \big(\bar{w}-\|\tilde{w}(t)\|\big) + \|\tilde{w}(t)\|^2.
\end{align*}
Since $\|\tilde{w}(t)\| \le \bar{w}$ for all $t \ge 0$, it follows that $h_5(t,x) -\ubar{h}_5(t,x) \geq 0$, which implies that $\ubar{h}_5(t,x)\leq h_5(t,x)$.
A similar procedure is used to show that $\ubar{h}_6(t,x) \leq h_6(t,x)$.
\end{proof}



\begin{proof}[Proof of \cref{prop:forward_invariance}]
Since $x_0 \in \ubar{\SC}(0)$, it follows from \cref{prop:robust_constraint} that $x_0 \in \SC(0)$. Moreover, since for all $t \in [0,\infty)$, $x(t) \in \ubar{\SP}(t)$, it follows from \cref{eq:P_bar_set} that for all $t\in[0,\infty)$, $x(t)\in\SP(t)$.
Next, it follows from \cite[Theorem 3]{xiao_high-order_2022} that for all $t \in [0, \infty)$, $x(t) \in \SC(t)$. Moreover, it follows from \cref{eq:C_set,eq:C_bar_set} that $x(t)\in\ubar{\SC}(t)$.
\end{proof}

\begin{proof}[Proof of \cref{prop:omega_feasibility}]
    We consider 2 cases: $x \in \mbox{bd } \SH(t)$, and $x \in \BBR^{n} \setminus \mbox{bd } \SH(t)$.
    First, let $x \in \mbox{bd } \SH(t)$. Let $\bar{\kappa}_1 \in \BBR$ and
    \begin{equation*}
        \bar{u}_1 = \frac{-\frac{\partial h(t,x)}{\partial t} - \frac{ \partial h(t,x)}{\partial x} f(x)}{\|\frac{\partial h(x,t)}{\partial x}B\|^2} \Big[\frac{\partial h(x,t)}{\partial x}B \Big]^\top.
    \end{equation*}
    Then, \cref{eq:nu} yields $\nu(t,x,\bar u_1,\bar \kappa_1) = 0$.
    Next, let $x \in \BBR^{n} \setminus \mbox{bd } \SH(t)$. Let $\bar{u}_2 \in \BBR^3$ and
    \begin{equation*}
        \bar{\kappa}_2 = \frac{-\frac{\partial h(t,x)}{\partial t} - \frac{ \partial h(t,x)}{\partial x} \left [ f(x) + B \bar{u}_2 \right ] - \alpha \big( h(t,x) \big )}{h(t,x)}.
    \end{equation*}
    Then, \cref{eq:nu} yields $\nu(t,x,\bar u_2,\bar \kappa_2) = 0$. 
\end{proof}

\bibliographystyle{ieeetr}
\bibliography{References}

\end{document}